\documentclass[acmsmall,screen,nonacm]{acmart}

\acmJournal{PACMPL}
\acmVolume{1}
\acmNumber{CONF} 
\acmArticle{1}
\acmYear{2018}
\acmMonth{1}
\acmDOI{} 
\startPage{1}

\setcopyright{none}

\usepackage{booktabs}   
\usepackage{subcaption} 

\usepackage[utf8x]{inputenc}

\usepackage{amssymb, upgreek}
\usepackage[noabbrev,capitalize]{cleveref}
\usepackage{bbm}
\usepackage{dsfont}
\usepackage{amsmath}
\usepackage{mathtools}
\usepackage{framed}

\usepackage{fancyvrb}
\usepackage{listings}
\usepackage{lstautogobble}

\usepackage{color}
\definecolor{keywordcolor}{rgb}{0.7, 0.1, 0.1}   
\definecolor{commentcolor}{rgb}{0.4, 0.4, 0.4}   
\definecolor{symbolcolor}{rgb}{0.0, 0.1, 0.6}    
\definecolor{sortcolor}{rgb}{0.1, 0.5, 0.1}      
\definecolor{errorcolor}{rgb}{1, 0, 0}           
\definecolor{stringcolor}{rgb}{0.5, 0.3, 0.2}    
\usepackage{tikz}
\usepackage{tikz-cd}
\usetikzlibrary{positioning, calc, shapes, backgrounds, fit,
matrix, trees, arrows, patterns, math, intersections,shapes.geometric}

\theoremstyle{plain}
\newtheorem{theorem}{Theorem}[section]
\newtheorem{lemma}[theorem]{Lemma}
\theoremstyle{definition}
\newtheorem{example}[theorem]{Example}
\newtheorem{definition}[theorem]{Definition}
\newcommand{\figlabel}[3]{
    \node[figlabeltextstyle] (l) at ($ #1 + #2 $) {#3};
}
\newcommand{\x}{\times}

\newcommand{\T}{\mathcal{T}}
\newcommand{\2}{\mathbbm{B}}

\newcommand{\sem}[1]{\llbracket #1 \rrbracket}
\newcommand{\semm}[1]{\llbracket #1 \rrbracket^1}
\newcommand{\semn}[1]{\llbracket #1 \rrbracket}
\newcommand{\semnn}[1]{\llbracket #1 \rrbracket^1}

\newcommand{\semnop}{\operatorname{eval}}
\newcommand{\rep}[2]{{\Uparrow_{#1}#2}}
\newcommand{\repop}{\operatorname{\Uparrow}}
\newcommand{\sig}[2]{\Sigma_{#1}#2}
\newcommand{\sigop}{\operatorname{\Sigma}}
\newcommand{\streamsig}{S\to I\x V \x \2\x S}

\newcommand{\numbers}[1]{[#1]}
\newcommand{\mseq}[1]{[1,2,\ldots,#1]}

\newcommand{\signature}{\mathds{T}}
\newcommand{\streams}[2]{\mathcal{S}(#1, #2)}

\renewcommand{\i}{\operatorname{index}}
\renewcommand{\v}{\operatorname{value}}
\newcommand{\ready}{\operatorname{ready}}
\newcommand{\skipfn}{\operatorname{skip}}
\newcommand{\skp}{\operatorname{skip}}

\renewcommand{\L}{\mathcal{L}}
\renewcommand{\S}{\mathcal{S}}
\renewcommand{\T}{\mathcal{T}}
\newcommand{\map}{\operatorname{map}}
\newcommand{\sorts}{\mathcal{S}}
\newcommand{\reaches}[2]{#1 \to^* #2}
\newcommand{\id}{\operatorname{id}}
\newcommand{\remark}{}
\newcommand{\notation}{\emph{Notation. }}

\definecolor{mypink1}{rgb}{0.158, 0.788, 0.478}
\definecolor{mypurp}{rgb}{0.50, 0.0, 1.0}
\definecolor{keywordcolor}{rgb}{0.50, 0.0, .5}
\definecolor{mygray}{rgb}{0.5,0.5,0.5}
\newcommand{\ttt}[1]{\texttt{#1}}
\newcommand{\size}{\operatorname{size}}
\newcommand{\sizez}{\operatorname{size}^0}
\newcommand{\ttG}{\ttt{G}}
\newcommand{\note}{\paragraph{Note}}

\begin{document}

\title[Correct Compilation of Semiring Contractions]{Correct Compilation of Semiring Contractions}


\author{Scott Kovach}
\affiliation{
  \institution{Stanford University}            
  \country{USA}                    
}
\email{dskovach@stanford.edu}          

\author{Fredrik Kjolstad}
\affiliation{
  \institution{Stanford University}           
  \country{USA}                   
}
\email{kjolstad@stanford.edu}         

\begin{abstract}
We introduce a formal operational semantics that describes the fused execution of variable contraction problems,
which compute indexed arithmetic over a semiring and generalize sparse and dense tensor algebra, relational algebra, and graph algorithms.
We prove that the model is correct with respect to a functional semantics.
We also develop a compiler for variable contraction expressions and show that its performance is equivalent to a state-of-the art sparse tensor algebra compiler, while providing greater generality and correctness guarantees.
\end{abstract}




\maketitle

\section{Introduction}

Scientific computing and data analysis are critically reliant on performant execution.
The drive toward optimization has led to diverse systems that handle each step of a data processing pipeline:
databases for efficient storage and access,
domain-specific applications for physics simulation and optimization,
and kernel libraries for fundamental numeric operations like sparse matrix multiplication.
However, building a computation modularly out of well-tested components does not always cut it.
Recent work on domain-specific data processing systems has demonstrated that \emph{fusing} irregular and sparse computations across tasks or operations can produce radical performance improvements:
relational query optimization techniques can be applied to linear algebra~\cite{aberger2018levelheaded},
the composition of gradient descent with a relational query can be jointly optimized to drastically reduce costs~\cite{LMFAO},
and arithmetic operations on sparse matrices can be asymptotically faster with correct fusion~\cite{kjolstad2017tensor}.

Although these systems achieve new heights of performance, they are complex.
Users in scientific and safety-critical engineering domains may not trust that generated code is correct without a formal model of behavior. 
There is a mature literature on certified compilation and verification of imperative programs~
\cite{softwaremodelchecking}
\cite{chlipala2013bedrock}
\cite{leroy2016compcert},
but limited work that generalizes across high-performance programming systems for sparse computation like those above.

We will show, however, that a large class of data processing problems can be expressed within a simple language built from operators that multiply and aggregate multi-dimensional arrays.
We call these \emph{variable contraction problems} (\cref{sec:language}).
Our primary contribution is an operational semantics for fused programs that solve contraction problems.
By disentangling issues that arise in sparse computation, our semantics allows us to prove that
a broad class of efficient fused programs calculate their intended results.
Our approach is highly parametrized:
it can make use of arbitrary (sparse) data representations and express problems such as tensor contractions, relational queries,
path algorithms, probabilistic inference, boolean satisfiability, and more (\cref{sec:eg}).

The approach is based on three key ideas:

\begin{itemize}

\item
  Practical solutions must support an extensible set of sparse and dense data formats.
  To support all such representations,
  we introduce \emph{indexed streams} (\cref{sec:stream_def}), which can model arbitrary stateful, sequential computations.

\item
We decompose the notation of contraction problems using a language of \emph{contraction expressions} formed using three operators.
We implement these operators on (finitary descriptions of) streams (\cref{sec:combinators}).
This naturally gives rise to a compilation method that produces code with predictable time and memory usage.

\item
Hierarchical data structures are typically used to represent sparse objects such as multi-dimensional arrays and database indices.
We model these using \emph{nested streams} (\cref{sec:nested_streams}), which emit other streams as values.

\end{itemize}

\begin{figure}[t]
    \centering
    \begin{tikzpicture}
        \tikzstyle{figlabeltextstyle} = [black!50,font=\footnotesize,inner sep=1pt]
        \tikzstyle{figlabellinestyle} = [black!0,line cap=round]
        \node (L) at (0  ,  0)    {\Large{$\mathcal{L}$}};
        \node (S) at (2, -1) {\Large{$\mathcal{S}$}};
        \node (T) at (2,  1) {\Large{$\mathcal{T}$}};
        \path[] (L) edge[->,-stealth] (T);
        \path[] (L) edge[->,-stealth] (S);
        \path[] (S) edge[->,line width=1,-stealth] (T);
        \figlabel{(L)}{(-0.39,-0.4)}{\Cref{sec:semantics},}
        \figlabel{(L)}{(-0.39,-0.60)}{\cref{def:term_algebra}}
        \figlabel{(T)}{(0, 0.4)}{\Cref{sec:semantics}}
        \figlabel{(S)}{(0,-0.4)}{\Cref{sec:streams}}
        \node[font=\footnotesize] (M) at (2.5,0) {$\sem{-}$};
        \figlabel{(M)}{(0.3,-0.4)}{\Cref{sec:stream_evaluation}}
    \end{tikzpicture}
    \vspace{1em}
    \caption{\label{fig:overview}
      An intuitive overview of the language we handle is given in \cref{sec:semantics}.
      Our main conceptual contributions are defined in \cref{sec:streams,sec:nested_streams}.
      The commutativity of the above diagram, relating stream semantics and functional semantics for variable contraction problems,
      is proved in \cref{sec:proofs}.
    }
    \Description{Paper Overview}
\end{figure}
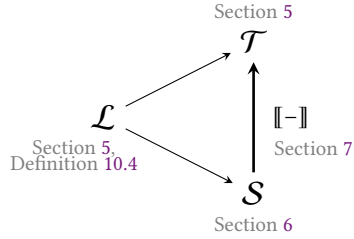

This approach is modular: rather than specifying and formalizing a monolithic compiler
from problems to streams, we specify a set of operators that can be used to algebraically
interpret a given problem as a stream. 
This set may be extended with new operators without disturbing pre-existing ones,
and new data formats are abstracted by the stream interface.

We apply the strategy to implement a prototype compiler.
As a result of its modular design, which directly aligns with the semantics definitions in \cref{sec:streams},
the implementation avoids the \emph{expression problem}~\cite{wadler-expression}: we can independently add new stream datatypes, new operators on streams,
and new compilation heuristics without modifying any pre-existing code.
Our implementation generates code that matches the performance of code generated by the state-of-the-art sparse tensor algebra compiler TACO
~\cite{kjolstad2017tensor}
on a subset of benchmarks, and it does so using two orders of magnitude less implementation code.

The correctness of the semantics is depicted in \cref{fig:overview}.
Here, $\L$ is a language for expressing contraction problems, 
$\T$ is the standard domain of functions,
and $\S$ is the domain of streams.
This commutative diagram relates the specification map $\L\to\T$, which interprets the operators on the domain of functions,
the map $\L\to \S$, which interprets the operators on the domain of streams, and the evaluation homomorphism $\S\to\T$.

Our primary contributions are
\begin{itemize}
\item
A high-level notation for specifying contraction problems (\cref{sec:semantics}),
\item
the \emph{indexed stream} representation of hierarchical, indexed data (\cref{sec:streams,sec:nested_streams}),
\item
a set of combinators which multiply and aggregate the values of indexed streams (\cref{sec:streams}),
\item
a proof that indexed stream combinators faithfully compute solutions to contraction problems (\cref{sec:proofs}), and
\item
a compiler that directly implements the stream model with performance comparable to the compiler TACO (\cref{sec:compilation}).
\end{itemize}

\section{Illustrative Examples}
\label{sec:eg}

The examples in this section are meant to illustrate the expressive range of the problem we consider before
formally defining it in \cref{sec:language}.
We do not address any details of efficient computation until \cref{sec:challenges}.

\begin{example}[Matrix Products]
  \label{eg:1}
  A matrix with $d_1$ rows, $d_2$ columns, and entries in a set $R$
  can be viewed as a function $I_1 \x I_2 \to R$, where $I_1 = \{1,2,\ldots,d_1\}$ and
  $I_2 = \{1,2,\ldots,d_2\}$.
  Supposing elements of $R$ can be added and multiplied, if we have another matrix $B : I_2\x I_3 \to R,$
  the \emph{matrix product} $AB$ is defined by the formula
  $$ (AB)(i_1,i_3) = \sum_{i_2\in I_2} A(i_1, i_2)B(i_2,i_3).$$

  \qed
\end{example}

\begin{example}[Relational Queries]
  \label{eg:2}
  \emph{Relational algebra} is concerned with the semantics of various operators that transform relations,
  especially 
  \emph{selection} ($\sigma_P$), \emph{projection} ($\pi_I$), and \emph{natural join} ($R\bowtie S$).

  A \emph{relation} $R$ between attribute sets $A,B$ is defined to be a subset of their cartesian product: $R\subseteq A\x B$.
  Equivalently, a relation is specified by an indicator function $R : A\x B\to \2$ to the two-element set $\2 = \{\bot,\top\}$.
  Here, $\bot$ denotes false and $\top$ denotes true;
  to obtain the corresponding subset of $A\x B$, take $R^{-1}(\top)$.
  Similarly, a unary predicate on $A$ is a function $A\to\2,$
  a binary predicate on $A\x B$ is a function $A\x B\to \2$, and so on for more attributes.

  For $u,v\in\2$ define $u+v = u\lor v$ and $u\cdot v = u\land v$.
  Using these operations and the functional point of view on relations,
  we can write an arbitrary relational algebra expression in a form similar to the matrix product above.

  For example, suppose we have relations $R : A\x B\to\2$ and $S : B\x C\to \2$ and predicates $P : C\to \2$ and $Q : A\x C\to\2.$
  The expression $e = \pi_{A}(\sigma_Q(\sigma_P(R\bowtie S)))$ is equivalent to $$e(a) = \sum_{b\in B}\sum_{c\in C}(R(a,b)\cdot S(b,c)\cdot P(c)\cdot Q(a,c)).$$
  That is, $e(a) = \top$ exactly when there exists a tuple $t \in R\bowtie S$ that satisfies $P$ and $Q$
  and has $\pi_A(t) = a$.

  This example highlights that the binary join $R\bowtie S : A\x B \x C \to \2$ is essentially an instance of matrix multiplication:
  $\pi_{AC}(R\bowtie S)(a,c) = \bigvee_b R(a,b)\land S(b,c) = \sum_b R(a,b)\cdot S(b,c).$
  \qed
\end{example}

\begin{example}[Path Operations]
  \label{eg:3}
  Suppose we have a finite graph $(V,E)$.
  An edge weighting can be represented as a function $w : V_1\x V_2\to \mathbb{R}\cup \{\infty\}$, where $V_1=V_2=V$.
  We assume $w(u,v)=\infty$ if there is no edge $(u,v)$ in the graph.

  Many iterative path-finding algorithms compute a frontier of shortest paths:
  given a source vertex $v_0$ and best-so-far shortest path length for each vertex
  organized as a vector $d : V_1\to R$, we compute a new distance vector
  $$d'(v_2) = \min_{v_1\in V_1} \left[d(v_1) + w(v_1, v_2)\right].$$
  This is the shortest way of reaching $v_2$ by extending an existing path with one edge.

  The $(\min, +)$ semiring is defined on the set $\mathbb{R}\cup\{\infty\}$:
  In this structure, \emph{addition} is given by $\min$ and \emph{multiplication} is given by addition (of extended real numbers).
  In this notation, the shortest path update expression is
  $$d'(v_2) = \sum_{v_1\in V_1} d(v_1)\cdot w(v_1, v_2).$$
  Moreover, we can obtain the shortest path itself using \emph{the same expression} by selecting a related semiring \cite{dolan2013fun}.

  \qed
\end{example}

\section{The Variable Contraction Problem}
\label{sec:language}

This section introduces a particular formulation of the problem of computing an aggregate over a
product of a collection of functions.
This problem was called the MPF (marginalize a product function) problem in \cite{aji2000generalized} and
FAQ-SS (Functional Aggregate Query, Single Semiring) in \cite{faq}.
It is an intuitively simple operation, but expressive enough to include many algorithms by varying the underlying \emph{semiring} $R$
and the functions involved.

\begin{definition}[Semiring]
A semiring is a set equipped with structures $(+,0)$ and $(\cdot, 1)$
satisfying the axioms of a commutative monoid and a monoid, respectively, and also the distributive and absorption laws:
  \begin{align*}
  x(y+z) &= xy+xz  \\
  (y+z)x &= yx+zx  \\
    0 &= 0\cdot x = x\cdot 0.
  \end{align*}
\end{definition}

\begin{example}[Semirings]
We give a few well-known examples of semirings:
\begin{itemize}
\item any ring or field with $(+,\cdot)$ is a semiring;
\item the set of booleans $\{\bot, \top\}$ under $(\lor,\land)$, and more generally any boolean algebra;
\item $\mathbb{R}\cup\{\infty\}$ with $(\min, +)$ (the \emph{tropical semiring}\cite{pin_taylor_atiyah_1998});
\item the set of square matrices with entries in a semiring under addition and matrix multiplication;
\item and finally, when $R$ is a semiring and $I$ is any set, the functions $I\to R$ form a semiring using pointwise addition and multiplication.
\end{itemize}
\end{example}

Most of the constructions in this work are parametrized by a semiring $R$,
an integer $m,$ and a sequence of (not necessarily distinct) finite, totally ordered sets $[I_1, I_2,\ldots, I_m]$.
We often identify the index $i$ and its corresponding set $I_i$.

Define $$\numbers{m} = \{1,2,\ldots,m\}.$$

\begin{definition}[Variables]
For any subset $S \subseteq \numbers{m}$,
$I_S$ denotes $\prod_{i\in S}I_i$, called an \emph{indexing set}.
We call a function $V : I_S \to R$ a \emph{variable} (a generalized quantity in $R$ that varies over its indexing set $I_S$).
The subset $S$ is the variable's \emph{shape}.
Let $\pi_{S} : I_{\numbers{m}} \to I_S$ be the projection function.
When $x\in I_{\numbers{m}} $, $V(x)$ means $V(\pi_{S}(x))$.

\end{definition}
\begin{definition}[Variable Contraction]
  \label{def:contraction}
  An instance of the \emph{variable contraction problem} is defined by
  a set of shapes
  $$\{S_k\subseteq\numbers{m}\mid 1\le k\le n\},$$
  a set of variables
  $$\{V_k : I_{S_k} \to R \mid 1\le k\le n\},$$
  and a subset of indices $C \subseteq \numbers{m}$.

  The indices in $C$ are said to be \emph{contracted} or \emph{marginalized}.
  The indices in the complement $F = \numbers{m}\setminus C$ are called \emph{free}.
  Note that for any tuples $x_F\in I_F$ and $x_C\in I_C$ we have a corresponding tuple $x_F\cup x_C\in I_{F\cup C} = I_{\numbers{m}}$.

  The problem is to compute a table of values for the function $V : I_F\to R$
  defined over the free indices by the equation

  \begin{align}
    V(x_F) = \sum_{x_C \in I_C} \prod_{1\le k\le n} V_k(x_F\cup x_C).
  \end{align}
\end{definition}

\begin{example}[\cref{eg:1}, part 2]
  Each of the examples given in the previous section is a variable contraction problem.
  For example, matrix multiplication is specified by $S_1 = \{1,2\}, S_2 = \{2,3\}, V_1 = A, V_2 = B,$ and $C = \{2\}$.
\end{example}

\note
Computations that make use of a bilinear form to combine two or more (representations of) \emph{tensors}
are often referred to as \emph{contractions}.
This terminology comes originally from differential geometry~\cite{ricci1900methodes}.
Since variable contraction can express traditional contraction
(once a particular basis is chosen), but we do not require that variables actually represent true tensors,
we consider this operation to be a sort of generalized contraction.

\section{Challenges of Efficient Execution}
\label{sec:challenges}

\begin{figure}[b]
\includegraphics[scale=0.4]{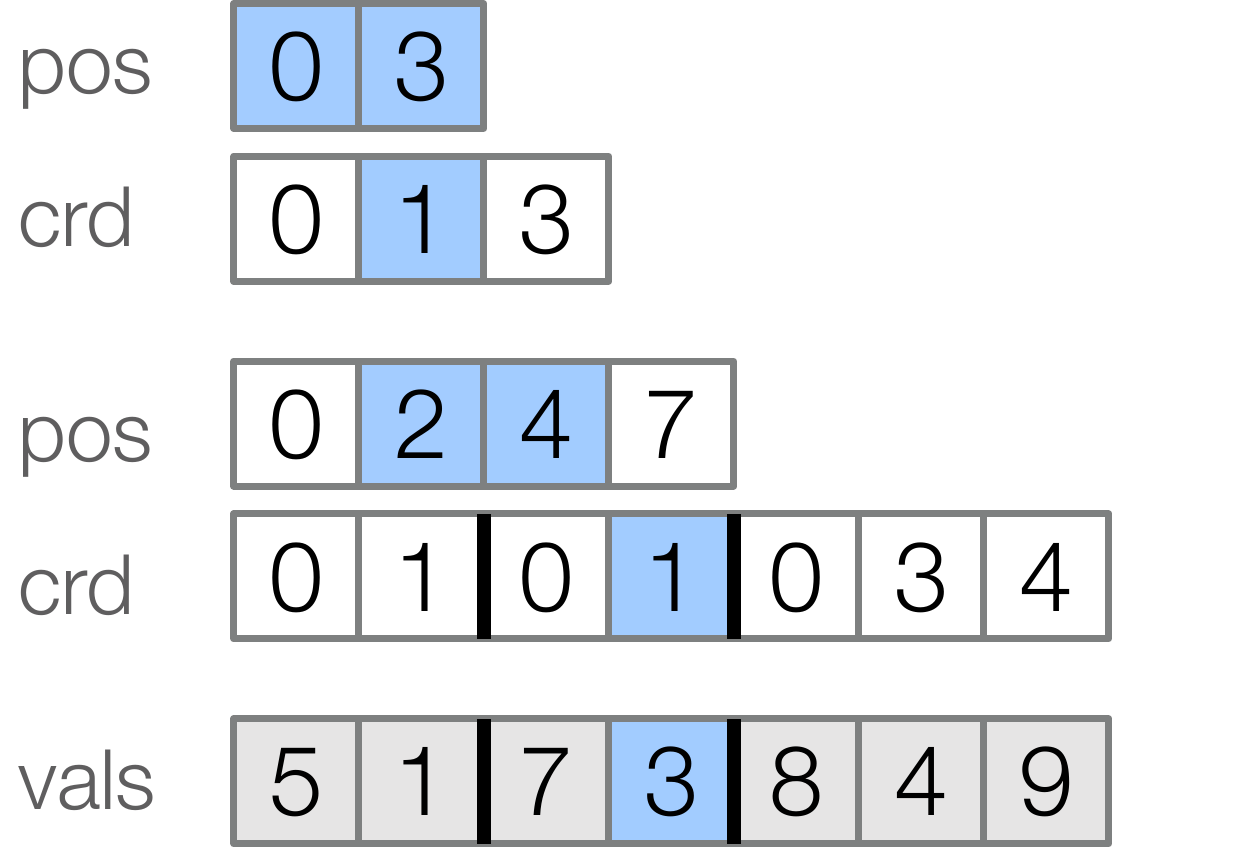}
\caption{The DCSR (doubly compressed sparse row) matrix format storing a matrix with three non-empty rows (0,1,3)
  and four non-empty columns (0,1,3,4).
  Each \ttt{crd} array stores a series of what we call index values.
  The entries of the matrix are stored in the flattened \ttt{vals} array.
  The highlighted trail depicts the entry $(1,1)\mapsto 3$.
\label{fig:dcsr}
}
\end{figure}

In this section, we discuss two concepts which significantly affect the performance of many variable contraction queries in reality:
fusion and hieararchical iteration.
Our semantics for algorithmic solutions to the contraction problem is motivated by these two core issues.
As a case study for the present work and for the sake of concrete examples, we focus on the domain of array programs.

By viewing an array as a function from its indexing set to the corresponding numeric entry,
the variable contraction problem encompasses the operation commonly known as ``tensor contraction'' or ``einsum''~\cite{harris2020array} in work on machine learning and numerical programming.
When inputs are represented using \emph{sparse data formats}, values that are known to be zero are not explicitly
represented in memory.
For example, the coordinate list representation stores a list of nonzero values and their associated coordinates, sorted by coordinate.
Although this data may be stored using an array, its semantic content is a vector in some higher dimensional space.
An even more highly compressed format is shown in \cref{fig:dcsr}; this data structure stores
each non-empty row coordinate once, and values are stored in a single cache-friendly array.

Array programs are especially sensitive to issues of memory hieararchy and locality of reference.
It is necessary in practice to minimize intermediate allocations and iterate over contiguous segments of data as much as possible
since sparse structures do not support constant-time random access.
We capture this issue using the notion of \emph{fusion}.
\begin{definition}
  \label{def:fusion}
A computation implementing variable contraction is \emph{fully fused} when the following conditions are met:
\begin{itemize}
\item
the number of memory locations modified during the span of time spent computing the $i^{th}$ output value is bounded by a constant depending only on the variable contraction expression, not the input data, and
\item
each data structure representing a variable with domain $I_S$ is only iterated by the lexicographic order on $I_S$, possibly more than once.
\end{itemize}
\end{definition}

\begin{example}
If the element-wise product of three vectors $V_1\cdot V_2\cdot V_3$
is computed by first storing the entries of $V' = V_1\cdot V_2$ into memory and then
computing all entries of $V'\cdot V_3$, the computation is \emph{not fused}:
to compute any entry of the output, we must first compute all of $V'$;
however, the number of memory locations modified while computing $V'$ depends on the length of $V_1$ or $V_2.$
In contrast, computing the values $V'(x)$ sequentially by $V_1(x)\cdot V_2(x)\cdot V_3(x)$ is fused.
\end{example}

Additionally, matrices and higher-dimensional arrays are often used to represent mathematical tensors for the sake of computation.
A matrix $A : I_1\x I_2\to R$ can be thought of as a two-level object:
each index $x_1\in I_1$ gives rise to a \emph{row} $A(x_1) : I_2\to R$.
When these objects are represented in a sparse format to be iterated, it is especially important to exploit this hierarchy
while performing multiplications.
Because multiplication satisfies $0\cdot x =0$, it is permissible to skip over any segment of data that is not present in
\emph{all} the input factors.
Skipping over an index at the highest level saves the work of operating on an entire slice of a tensor.
Notice that this very same phenomenon arises in database traversals which use table-indexes to skip over absent tuples
(tuples which map to $0 =\bot$, according to the interpretation given in \cref{eg:2}).

Because memory locality and demand-driven, hierarchical computation arise in so many computational contexts,
we encode them directly into our semantic domain.
We will address these issues starting in \cref{sec:streams} by defining nested streams
and operators which multiply and aggregate streams directly.
These operators guarantee fusion and enable various approaches to hieararchical iteration with skipping.

\section{A Functional Interpretation}
\label{sec:semantics}

To motivate our operators on streams, we first introduce them for variables.
This implementation will later serve as our specification of intended stream behavior.

We break the variable contraction problem down into three operators which act on singular functions or pairs of functions.
We can obtain the result for any instance of the contraction problem by applying these operators to a given set of variables.

\subsection{The Contraction Operators}
  \notation
When $S\subseteq T\subseteq\numbers{m}$, there is a projection operator $\pi_S : I_T \to I_S$ which projects a tuple onto
the smaller indexing set.
In particular, for $i\in S$, we introduce the notation $\pi_{i} = \pi_{S\setminus\{i\}}.$

\begin{description}
\item[Replication]
Whenever $S\subseteq T$, it is possible to redefine a variable of shape $S$ to have shape $T$.

Explicitly, the \emph{replication} operator, $\repop_{i},$ is defined for $i\notin S$.
This transforms a function $V : I_{S}\to R$ to type $I_{S\cup\{i\}}\to R$ by the rule
$$(\rep{i}{V})(x) = V(\pi_{i}(x)).$$

By applying $\repop_{i}$ to $V$ for each $i\in T\setminus S,$ we obtain a function of type $I_T\to R$.
Since the order of projections is irrelevant the final function only depends on $T$ and $V$.

\item[Multiplication]
Since $R$ is a semiring, functions defined on a common domain can be multiplied pointwise:
given $V_1, V_2 : I_S \to R$, define $(V_1\cdot V_2)(x) = V_1(x)\cdot V_2(x).$

\item[Summation]
  Variables can also be \emph{added} pointwise in the same way.
When $V : I_S\to R$, $i\in S$, and $x_i\in I_i$, $V(x_i)$ denotes the partial application of $V$, which
has type $I_{S\setminus\{i\}}\to R$.

Given $V : I_{S} \to R$ and $i \in S$, define $\sig{i}{V} : I_{S\setminus\{i\}}\to R$ by
$$(\sig{i}{V}) = \sum_{x_i\in I_i}V(x_i).$$

\end{description}

\subsection{Equivalence}
Given these three operators, it is straightforward to construct the result of a contraction problem:
\begin{enumerate}
\item For each variable, apply replications (in any order) to obtain a new variable of type $I\to R$.
\item Multiply the variables pairwise in their given order.
\item Apply summation (in any order) for each index in the set $C$.
\end{enumerate}

\begin{example}[\cref{eg:1}, part 3]
  Given variables $A : I_1\x I_2 \to R$ and $B : I_2\x I_3 \to R$, their matrix product is
  $$ \sig{2}{((\rep{3}{A})\cdot(\rep{1}{B}))}.$$
\end{example}

\remark{}
Note that the product operation is associative (and commutative, if $R$ is commutative),
and that it distributes over summation since $R$ is a semiring:
$f\cdot\sig{i}{g} = \sig{i}{(f\cdot g)}.$

\section{A Streaming Interpretation}
\label{sec:streams}
In this section, we define \emph{indexed streams},
which model the computation of a sparse array one element at a time.
In \cref{sec:combinators} we will redefine the three contraction operators on the domain of streams
and show that they construct efficient, fused computations.

We also define a natural semantics function which maps back into the domain of functions.
In \cref{sec:proofs} we show that this map is a \emph{homomorphism};
hence streams correctly model variables and any algebraic optimizations performed on them are sound.

\subsection{Indexed Streams}
\label{sec:stream_def}
\begin{definition}[Indexed Streams, $\streams{I}{V}$]
Given sets $I$ and $V$, 
an \emph{indexed stream of type $I\to V$} is a tuple
$$(S, q, \i, \v, \ready, \delta)$$
with $S$ the \emph{state space} and $q\in S$ its \emph{state}.
The remaining elements are functions with the following names and types:

\begin{align*}
 \i     : S \to I,  & \textrm{ the \emph{index function}} \\
 \v     : S \to V,  & \textrm{ the \emph{value function}} \\
 \ready : S \to \2, & \textrm{ the \emph{ready function}} \\
 \delta : S \to S, & \textrm{ the \emph{successor function}} \\
\end{align*}

The set of streams of type $I\to V$ is written $\streams{I}{V}$.

When notationally desirable, the four functions can be packaged into one:
$$ f : \streamsig .$$

When clear from context, a stream $(S, q, f)$ will be simply referred to by its state $q$.
In particular, if $q = (S, q, f)$ is a stream, then $\delta(q)$ is the stream $(S, \delta(q), f).$
Lower-case variables $q,r,s,\ldots$ and $a,b,\ldots$ are used to denote streams (or stream states).

\end{definition}

A stream should be viewed as a transition system: $\delta(q)$ is the state that {follows}
$q$, and $(\i(q),\v(q))$ are to be thought of as data output by the stream in state $q$.
As long as a stream has a finite set of reachable states, it is possible to evaluate it to obtain a variable.
First we formalize the notion of finiteness, then describe the evaluation procedure.

\begin{definition}[Finite Streams]
A state $r$ is \emph{reachable from $q$}, written $\reaches{q}{r}$, if $r = \delta^k(q), k\ge 0.$
A state $r$ such that $r = \delta(r)$ is called \emph{terminal}.
If the set of states reachable from $q$ contains a terminal state then it is necessarily finite,
and we say that the stream $q$ is {finite}.

\end{definition}

Given a function $f : V\to R$, there is a natural way of interpreting a stream of type $I\to V$ as a variable of type $I\to R$.
First note that
the space $I\to R$ is spanned by the following simple functions:

\begin{definition}
For $x\in I$ and $v\in R$, let
$x\mapsto v$
denote the function
$$ (x\mapsto v)(y) =
\begin{cases*}
    v & if $x = y$ \\
    0 & otherwise.
\end{cases*}
$$
\end{definition}

We can evaluate a finite stream $q\in\streams{I}{V}$ by summing a series of terms computed over its set of reachable states.
For a given state $r$, if the state is ready ($\ready(r)=\top$), the term is $\i(r)\mapsto f(\v(r))$;
otherwise, if $\ready(r)=\bot$, the term is zero.
The following definition formalizes this.

\begin{definition}[Stream Evaluation: $\sem{-}_f$]
  \label{def:eval-v1}
  Suppose $q$ is a finite stream of type $I\to V$ and $f : V\to R$ is a function which maps
  stream values into a semiring.
  Let $$\semm{q}_f =
  \begin{cases*}
    0 & $\ready q = \bot$ \\
    \i(q)\mapsto f(\v(q)) & otherwise.
  \end{cases*}$$
  Then $\sem{q}_f : I\to R$ is the function
  \begin{align*}
    \sem{q}_f = \sum_{\reaches{q}{r}} \semm{r}_f
  \end{align*}

  We will usually omit the function $f$ when it is arbitrary or understood from context.

Stream evaluation satisfies the following \textbf{important identity:}
If $\delta(q) = q$, then $\sem{q} = \semm{q}$;
otherwise, $$\sem{q} = \semm{q} + \sem{\delta(q)}.$$

\end{definition}

\begin{example}
  \label{eg:sparse-vec}
  Suppose we have a variable $V : I_1 \to R$ over a single index and that $V(x)$ is nonzero for exactly $k$ values of $x$.
  This means that there is a sequence $[(x_1, v_1), (x_2, v_2),\ldots, (x_k, v_k)]$ such that $V(x_\ell) = v_\ell$ and $V(x)=0$ for all $x$ not appearing in the sequence.
  With this data, we can construct a stream $q\in\streams{I_1}{R}$ such that $\sem{q}_{id} = V$:
  $$ (S = \numbers{k}, q = 1, \i(r) = x_{r},\v(r) = v_{r},\ready(r) = \top, \delta(r) = r+1).$$
  Note that for this definition and elsewhere, we use saturating addition on $\numbers{k}$ so that $k+1 = k$ is a terminal state.
\end{example}

\remark
This definition captures the key intuition for stream evaluation.
However, to fully solve the contraction problem we will make use of \emph{nested streams}.
For example, to model a two-dimensional variable defined on $I_1\x I_2$,
we use a stream from the set $\streams{I_1}{\streams{I_2}{R}}.$
By selecting an appropriate function $f,$ it is possible to reuse the definition of $\sem{-}_f$ as stated for nested streams.
We will describe this construction in \cref{sec:nested_streams}.

\subsection{Stream Combinators}
\label{sec:combinators}
Here we define the operators on streams that suffice to solve variable contraction problems.
The correctness of these operators is shown in \cref{sec:proofs}.

\subsubsection{Replication}
Replication is used to make a single value available at multiple states across time.
Computationally, it does not necessarily require copying or recomputing the value;
it simply makes a single object available at multiple times.

Replication is implemented via a constant stream:
\begin{definition}
  Given a bijective, order-preserving function $\i : \numbers{k}\to I$ and a value $v\in V$,
  the constant stream $\repop(v) \in\streams{I}{V}$ is given by
  $$ (S = \numbers{k}, q = 1, \i, \v(r) = v, \ready(r) = \top, \delta(r) = r+1). $$
  That is, the stream is always ready, always returns the value $v$, and iterates across the indices of $I$.
\end{definition}
\subsubsection{Multiplication}
\label{sec:multiplication}
Multiplication is used to combine the values of two streams.
Since we assume that $0\cdot x = x\cdot 0 = 0$, our operator performs the essential \emph{intersection optimization}:
it does not produce output at a given state unless both input streams are non-zero there.
This is implemented using a successor function that is similar to the familiar merge step of merge-sort.

\notation
Whenever $q$ is a stream, $S_q$ denotes its state space.
We define a pre-ordering on arbitrary streams $a,b\in\streams{I}{V}$ as follows:
\begin{align*}
a \le b := \i(a) < \i(b) \lor (\i(a) &= \i(b) \land \ready(a) = \bot) \\
a < b := \i(a) < \i(b) \lor (\i(a) &= \i(b) \land \ready(a) = \bot \land \ready(b)=\top) \\
\end{align*}

\begin{definition}
\label{def:multiplication}
Given streams $a, b : I \to V$ and a product operation $(\cdot) : V\x V \to V$,
the \emph{product stream} $a\cdot b$ over the statespace
$S_a\x S_b$ is given by
\begin{align*}
  \i(a, b) &= \max(\i(a), \i(b)) \\
  \v(a, b) &= \v(a)\cdot \v(b) \\
  \ready(a,b) &= \ready(a)\land\ready(b)\land \i(a) = \i(b) \\
  \delta(a,b) &=
                \begin{cases*}
                  (\delta(a), b) & if $a\le b$ \\
                  (a, \delta(b)) & otherwise.
                \end{cases*}
\end{align*}
\end{definition}

\subsubsection{Summation}
Summation aggregates the values of a stream across its set of reachable states.
First we define an addition operator for pairs of streams that is similar to the product stream construction;
unlike product, it iterates over the union of the non-zero values of either stream.
\begin{definition}
Given streams $a, b : I \to V$ and a sum operation $(+) : V\x V \to V$,
the \emph{sum stream} $a+ b$ over the statespace
$S_a\x S_b$ is given by
\begin{align*}
  \i(a, b) &= \min(\i(a), \i(b)) \\
  \v(a, b) &= (\i(a) = \i(a,b))\cdot\v(a) \\ &+ (\i(b) = \i(a,b))\cdot\v(b) \\
  \ready(a,b) &= \ready(a)\lor\ready(b) \\
  \delta(a,b) &=
                \begin{cases*}
                  (\delta(a), b) & if $a < b$ \\
                  (a, \delta(b)) & if $b < a$ \\
                  (\delta(a), \delta(b)) & otherwise.
                \end{cases*}
\end{align*}
\end{definition}

Our summation operator for streams is simple:
\begin{definition}[$\sigop$]
\label{def:summation1}
Suppose $V$ is a set with addition defined.
Given $q \in \streams{I}{V}$, define $\sigop q\in V:$
$$\sigop q = \sum_{\reaches{q}{r}}\v(r).$$
\end{definition}

So, a stream $q\in\streams{I}{R}$ is mapped to an element of $R$ and
a stream in $q\in\streams{I}{\streams{J}{V}}$ is mapped to $\streams{J}{V}$ via stream addition.

\section{Nested Stream Evaluation}
\label{sec:nested_streams}
In this section and the following one, we will build intuition for streams and the stream combinators by discussing their runtime behavior and several optimizations.
In order to discuss performance, we will first need to define the more general evaluation map for nested streams.

\label{sec:stream_evaluation}
To manipulate nested streams, we need to be able to formally manipulate their sequence of indices.
To facilitate this we introduce new notation for $\S$.

Suppose $\alpha$ is a subsequence of $\mseq{m}$.
We will use the standard list notations:
$[]$ denotes the empty list and $i :: \beta$ denotes the list with head $i$ and tail $\beta \subseteq \mseq{m}.$
Since the lists we consider are subsequences of $\mseq{m}$, for all $i'\in\beta,$ $i' > i$.

In analogy with variables, the simplest definition of the set of streams
that model variables of type $I_S\to R$ would be $\streams{I_S}{R}$.
However, hierarchical iteration is more natural using the following definition.

\begin{definition}[Nested Streams, $\S_\alpha$]
  \label{def:nested_streams}
Define
$$\S_{[]} = R, \; \text{and}$$
$$\S_{i::\alpha} = \streams{I_i}{\S_{\alpha}}.$$
\end{definition}

A nested stream produces a new stream at every state.
Thus, the obvious way to define its evaluation would be to recursively evaluate the yielded stream before proceeding to the next state in the sum.
This is straightforward to formalize using the earlier \cref{def:eval-v1} by specializing the $f$ function:

\begin{definition}[Stream Evaluation: $\semn{-}_\alpha$]
  \label{def:stream_eval_final}
Suppose $q\in \S_\alpha$.

For $\alpha = []$ (that is, $q\in R$), define $$\semn{q}_\alpha = q.$$ 

Recursively, for $\alpha = i :: \beta$, $$\semn{q}_\alpha = \sem{q}_{\semn{-}_\beta}.$$

This map sends a stream $q\in\S_\alpha$ to the variable $\semn{q}_\alpha : I_\alpha \to R$.
Notices that, in the recursive case, the map that is applied to stream values is itself $\semn{-}_\beta$.
The previous definition for $\sem{-}_f,$ which requires a map from stream values to a semiring $R$, is valid here
because functions $I_\beta\to R$ \emph{themselves} form a semiring whenever $R$ does.

From now on, since we exclusively work with nested streams, we use $\semn{-}$ to refer to $\semn{-}_\alpha$.

\end{definition}

Finally, we extend the replication and summation operators of the previous section to act on nested streams in a natural way.
Replication or summation can be applied to an arbitrary index $i$ of a function $I_S\to R$, and we can make the same generalization for streams.

\begin{definition}[Stream map]
For any function $f : A\to B$ there is a function
$$\map_i f : \streams{I_i}{A}\to\streams{I_i}{B}$$
  given by
$$(S_q, q, \i, \v, \ready, \delta)\mapsto (S_q, q, \i, f\circ \v, \ready, \delta).$$
This simply applies the function to each value produced by the stream.

We can also iterate this to map over an entire prefix of indices:
  \begin{align*}
  \map_{[]} &= \id \\
  \map_{i::\alpha} &= \map_i\circ\map_\alpha
  \end{align*}
\end{definition}

\begin{definition}[Nested Stream Operators]
  \label{def:nested_operators}
Suppose $\gamma = \alpha+[i]+\beta$ is an ordered sequence of indices.
The nested summation and replication operators are
$$\sigop_{i} : \S_{\gamma} \to \S_{\alpha+\beta}= \map_{\alpha}(\sigop)$$
and
$$\repop_{i} : \S_{\alpha+\beta}\to\S_{\gamma} = \map_{\alpha}(\repop).$$

The multiplication operator is already valid on nested streams:
it is defined whenever the value type can be multiplied, and
since it defines a multiplication for streams, this carries on inductively to streams in $\S_\alpha$ for all $\alpha$.
\end{definition}

With these definitions in hand, our stream operators are just as expressive as the variable operators.
Eventually we will show that the stream operators are also sound: $\semn{-}$ commutes with each operator.
In the following section, however, we will first discuss their usefulness in modeling optimizations.

\section{Performance Analysis}
\label{sec:performance-analysis}

We now have the tools to analyze the performance of several interesting streams.
Using the example of matrix multiplication, we illustrate that asymptotic runtime depends on the index ordering,
which has previously been demonstrated for
tensor contractions~\cite{kjolstad2020sparse,ahrens2022}.
Next, we give a small extension to the stream model to allow for logarithmic-time \emph{index skipping} that is pervasive in efficient data-processing algorithms.
We show that stream-based evaluation of relational queries 
can implement the worst-case optimal multiway join~\cite{ngo2018worst}.
These examples illustrate the range of behaviors that arise from streams built out of our combinators.
They also highlight the notational efficiency that the contraction language delivers.

First we formalize what we mean by runtime:

\begin{definition}[Stream Size: $\size$ and $\size^0$]
 Define the size of a value $v\in R$ to be one.
 Otherwise,
 $$\size(q) = \sum_{\reaches{q}{r}}\size(r).$$
 Also, let $\size^0(q)$ denote the number of reachable states: $|\{r \mid \reaches{q}{r}\}|.$
\end{definition}

For a stream $q$, $\size(q)$ is exactly the number of invocations of $\delta$ one would need to reach all the non-zero terms of $\semn{q}$.
Since streams are useful to model computations that perform a constant amount of work at each state
and $\semn{-}$ models the evaluation of a stream,
the size is a concrete and precise measure of practical runtime performance.
We note a few properties of size:

  \begin{itemize}
  \item A sparse matrix or table with $k$ non-zero entries can be represented by a stream with size exactly $k$ (in many possible ways).
  \item If $a,b \in \S$, then $\size^0(a\cdot b) \le \size^0(a) + \size^0(b)$.
  \item 
    For $q\in \S_{i::\alpha}$,
    \begin{align}
      \label{eq:size_bound}
      \size(q) \le \size^0(q)\max_{\reaches{q}{r}}(\size{r}).
    \end{align}
  \end{itemize}

The stream combinator definitions are parametrized by an index ordering.
The next example illustrates that, when we have flexibility to choose this ordering, some choices may achieve asymptotically better performance.

\begin{example}[\cref{eg:1}, part 4]
  Two index ordering strategies are the \emph{inner-product} method and the \emph{linear combination of rows}.
  For the inner product, we have two matrices $A : I_1 \x I_3$ and $B : I_2 \x I_3$.
  The contracted index is the innermost (last) one, and the contraction expression is
  \begin{align}
    \label{eq:inner_product}
    e_1 = \sig{3}{~(\rep{2}{A})\cdot(\rep{1}{B})}.
  \end{align}
  The reason for the algorithm's name is that iteration proceeds across $I_1\x I_2$, the output shape,
  and computes an inner-product of the corresponding row and column of $A$ and $B$ across $I_3$.

  For linear combination of rows, on the other hand, we compute on the two matrices $A : I_1 \x I_2$ and $B : I_2 \x I_3$.
  The expression (given in the earlier example) is
  \begin{align}
    \label{eq:linear_combination}
    e_2 = \sig{2}{~(\rep{3}{A})\cdot(\rep{1}{B})}.
  \end{align}
  This algorithm is called linear combination of rows because each row $A(i_1)$ is used to select a subset of the rows of $B$;
  the resulting row of output is a linear combination of these rows, weighted by corresponding non-zero entries of $A(i_1)$.

  Suppose that we have \emph{streams} $A : \S_{[1,3]}$ and $B : \S_{[2,3]}$ representing sparse matrices.
  Since the inner-product expression replicates the streams respectively across index 2 and 1,
  the entire cartesian product of non-empty rows of $A$ and $B$ is visited.

  \cref{eq:linear_combination}, on the other hand, is applied to streams $A : \S_{[1,2]}$ and $B : \S_{[2,3]}$.
  Replication is performed over index 1 and index \emph{3}.
  Multiplication intersects along $I_2$, which indexes the columns of $A$ and rows of $B$.
  This intersection is likely much smaller than $I_2$ in the inner-product case.

  In the simple case where a sparse matrix has $O(n)$ non-empty rows and $O(k)$ non-empty values within a row,
  we can use \cref{eq:size_bound} to obtain the conservative bounds $$\size(e_1) \in O(n^2k) \;\;\, \text{and} \;\, \size(e_2) \in O(nk^2).$$
  \qed
\end{example}

Finally, we address a defect in the bound on multiplication.
Earlier we noted the following weak bound:
for $a,b \in \S$, $\size^0(a\cdot b) \le \size^0(a) + \size^0(b)$.
The multiplication stream traverses one state at a time; in a \emph{skewed} data instance,
where the stream $a$ is much smaller than $b$, the multiplied stream may still need to spend much time traversing unnecessary elements of $b$.
However, data is often organized in some lexicographically sorted data structure,
so there are straightforward mechanisms for quickly skipping to a given index value.
B-tree indices, dense array lookup and galloping binary search are all such methods.

\begin{definition}[Searchable Streams]
  Let $q = (S, q, f)\in \streams{I}{V}$ be a stream and define $S\x_{\le}I = \{(s, x)\in S\x I\mid \i(s)\le x\}.$
  We say that $q$ is \emph{searchable} if there is a function $\skipfn : S\x_{\le}I \to S$ with the following properties:
  $$\i(\skipfn(q,x)) \ge x,$$
  $$\skipfn(q,x) = \delta^k(q) \textrm{ for some } k\ge 0,\  \text{and}$$
  $$\forall j, 0\le j < k, \i(\delta^j(q)) < x.$$
\end{definition}

\begin{example}
  When applied to searchable streams, the stream combinators all support efficient skip functions.
  For multiplication, take
  $$\skp(a\cdot b, x) = \textrm{let } a' := \skp(a,x)\textrm{ in } a' \cdot \skp(b,\i(a')).$$

  For replication over $I$, states are related to the indexing set by a bijection $\i : \numbers{k}\to I$,
  hence
  $$\skp(a, x) = \i^{-1}(x).$$
  Note that, in practice, this allows replicated streams to advance to a given index in constant time.

  For summation, the skip function follows from a skip function for ordinary binary summation,
  which is similar to  the multiplication case:
  $$\skp(a+b, x) = \skp(a,x) + \skp(b,x).$$

  Primitive dense streams (those that are backed by dense arrays) or implicit streams
  (those that are backed by a constant-time computable function) allow for constant time skipping,
  as in the example of replication.
  More generally, primitive streams used in practice allow for skip implementations that run in time bounded by the logarithm of the data size.
  \qed
\end{example}

The primary reason to discuss skip functions is for the sake of optimizing the stream product:
$$
  \delta(a\cdot b)_s =
                \begin{cases*}
                  \skp(\delta(a), \i(b))\cdot b & if $a\le b$ \\
                  a \cdot \skp(\delta(b), \i(a)) & otherwise.
                \end{cases*}
$$

Supposing $\size^0(a)< \size^0(b)$, using this function may asymptotically reduce the work done in multiplicative stream evaluation from
$\sizez(a)+\sizez(b)$ to $O(\sizez(a)\cdot\log\sizez(b)).$
For $m\ge 1$ and streams $\{q_i\}_{i\in\numbers{m}}$, the size of the product is bounded by
$$\sizez\left(\prod_iq_i\right)\in\tilde{O}\left(\min_{i\in\numbers{m}}\sizez(q_i)\right).$$
This building block is sufficient to implement an instance of
Generic-Join, Algorithm 3 from~\cite{ngo2014skew}.
The fact that nested stream evaluation is an instance of generic join is easy to check:
\begin{itemize}
  \item The global index ordering determines the choice of $I$.
  \item The nested loops of stream evaluation implement the recursive calls to \ttt{Generic-Join}.
  \item Relations that are broadcast over a given dimension have a constant-time transition, and
  \item for the set of other relations that are defined over a given index, $\delta_s$ implements an adequately efficient $m$-way sorted merge
    (assuming input relations are stored as tries).
\end{itemize}

The proofs of correctness we develop later generalize in a simple way to searchable streams.
However, for clarity, we do not handle searchable streams explicitly.

\section{Implementation}
\label{sec:compilation}

We implement the contraction operators as an embedded domain-specific language (DSL) in Lean~\cite{moura2015lean}
and a compiler, which we call \ttt{Etch}\footnote{\url{https://github.com/kovach/etch/}}.
The figures in this section depict unmodified Lean source code implementing some of the features of our implementation.
We demonstrate that the stream combinator concepts are adequate to construct
a real, working compiler that matches performance of hand-written code for sparse matrix computations.
We establish this by comparing the generated code to that generated by TACO \cite{kjolstad2017tensor}.

By implementing in Lean, our compiler is amenable to mechanized verification, but we leave this mechanization as future work.
Our prototype implementation demonstrates that
\begin{itemize}
\item
variable contraction problems can be specified in a
high-level, richly typed input language that helps prevent programmer error, and
\item high-performance programs can be generated from this notation using a concise, modular compilation approach.
\end{itemize}

We describe how to encode a stream as an imperative program,
how the compiler is structured out of loosely coupled components,
the input notation,
and our evaluation of generated code.

\subsection{Translating Streams to Efficient Imperative Code}

\begin{figure}[t]
\begin{minipage}[b]{0.31\textwidth}
\begin{lstlisting}
structure G (ι α : Type) :=
  (index : ι)
  (value : α)
  (ready : E)
  (valid : E)
  (init : Prog)
  (next : Prog)
\end{lstlisting}
    \caption{
      The fields \texttt{index}, \ttt{value}, \ttt{ready}, and \texttt{next} correspond to syntactic representations of the corresponding stream functions.
      The extra fields \ttt{init} and \ttt{valid} are added because states are represented implicitly as program state.
    \label{fig:lean-gen}}
\end{minipage}
\qquad
\begin{minipage}[b]{0.60\textwidth}
\begin{lstlisting}
instance base.eval : Ev ((E → Prog)) E :=
{ eval := λ acc v, acc v }

instance unit.eval [Ev α β] : Ev α (G unit β) :=
{ eval := λ acc v,
    v.init; Prog.while v.valid
      (Prog.if1 v.ready (Ev.eval acc v.value) ; v.next) }

instance level.eval  [Ev α β] : Ev (E → α) (G E β) :=
{ eval := λ acc v,
    v.init; Prog.while v.valid
      (Prog.if1 v.ready (Ev.eval (acc v.index) v.value);
      v.next) }
\end{lstlisting}
    \caption{
      Our nested stream evaluation function $\semnop$ is implemented using three cases that correspond to
      the base case $\S_{[]}$, the case of a contracted index, and the inductive case $\S_{i::\alpha}$.
      These functions recursively generate loop nests derived from the input stream.
    \label{fig:lean-loop}
    }
\end{minipage}
\end{figure}

Our objective is to translate an expression of the variable contraction language into imperative code.
Furthermore, the space and time usage for the generated code must match the performance analysis for evaluation given in \Cref{sec:performance-analysis}.
Recall that a stream is characterized by a transition function $f : \streamsig.$
Inspired by ~\cite{kiselyov2017stream}, we identify the statespace $S$ with the state of the imperative program.
Since states are no longer first-class objects, the $\delta$ function
becomes a procedure run for the sake of its side effects.

Output programs are represented in a simple Imp-style \cite{pierce2010software} imperative language fragment
(Lean type \ttt{Prog}) that can be directly transliterated to C and compiled.
We assume that index types $I_1,\ldots, I_m$, the semiring type $R,$ and boolean type $\2$ can be represented at run time.
They are represented syntactically with an expression type \ttt{E}; at run-time,
the expression evaluates to a value that varies based on present state.
With these substitutions, the resulting type for streams is given in \cref{fig:lean-gen},
which depicts our \emph{stream intermediate representation} (stream IR) called \ttG{}.

Since states are no longer first class objects, the new \ttt{init} fragment loads the initial state of a stream into memory.
Since we cannot directly compare states to detect a terminal state, we introduce the \ttt{valid} expression,
which is true exactly when a stream is in a non-terminal state.
These are the only extra components necessary for this imperative translation of the stream concept.

\subsection{Compiler Organization}

\begin{figure}[b]
\begin{minipage}{0.35\textwidth}
  {\footnotesize
\begin{align*}
  \i(a, b) &= \max(\i(a), \i(b)) \\
  \v(a, b) &= \v(a)\cdot \v(b) \\
  \ready(a,b) &= \ready(a)\land\ready(b) \\
              & \land \i(a) = \i(b) \\
  \delta(a,b) &=
                \begin{cases*}
                  (\delta(a), b) & if $a\le b$ \\
                  (a, \delta(b)) & otherwise
                \end{cases*}
\end{align*}
}%
\end{minipage}
\begin{minipage}{0.55\textwidth}
\begin{lstlisting}
def mul [has_hmul α β γ] (a : G E α) (b : G E β)
        : G E γ := {
  index := BinOp.max a.index b.index,
  value := a.value ⋆ b.value,
  ready := a.ready && b.ready && a.index == b.index,
  next  := Prog.if
             (a.index < b.index ||
             (a.index == b.index && a.ready.not))
                  a.next
                  b.next,
  valid := a.valid && b.valid,
  init  := a.init; b.init }

instance [has_hmul α β γ] : has_hmul
  (G E α) (G E β) (G E γ) := ⟨mul⟩
\end{lstlisting}
\end{minipage}
\caption{
  Our implementation of multiplication for code-generating stream objects (right).
  The definition of stream multiplication given earlier is reproduced on the left for comparison.
  This definition generalizes to arbitrary nested streams via typeclass resolution using the declaration shown below the definition.
  \label{fig:lean-mul}
}
\end{figure}

The essential feature of our semantics is that it is compositional:
each building block of variable contraction expressions is implemented as an operator on streams.
We preserve this essence in the design of the compiler.
Using the stream IR \ttG{} just described, the stream combinator definitions can be  mechanically reinterpreted as a compiler
by translating each one to act on \ttG{} objects.
That is, each operator is implemented as a function that takes one or two stream IR objects and returns another.
Moreover, nested streams are handled by simply nesting this datatype: for example, a two-level matrix inhabits \ttt{G E (G E E)}.
\cref{fig:lean-mul} demonstrates the ease of compilation by showing the complete implementation of stream multiplication.
The fields within \ttt{\{...\}} define the components of the product of input streams \ttt{a,b}.
The definition is valid for all \ttG{} objects with value types that can be multiplied;
hence, inductively, it defines multiplication of arbitrary nested streams
The typeclass instance shown below the definition enables the Lean elaborator to automatically derive appropriate multiplication operations.

The compilation strategy avoids the classic expression problem: the issue of simultaneously extending
a language with new types and new functions or methods on those types.
In our approach, new stream types are encoded as \emph{data:} each one is a new definition of type $G~\iota~\alpha$.
New methods are defined as operators on streams, and methods are composed primarily via typeclass inference.
This method of  composition encourages experimentation and makes the system loosely coupled:
entire definitions can be deleted, and only those programs which make use of them fail to compile,
while others remain well-defined.

\subsection{Front End}

\begin{figure}[b]
\begin{minipage}[c]{0.15\textwidth}
  {\footnotesize
\begin{align*}
  \textrm{mmul1:} &\sum_k A_{ij}B_{jk} \\
  \textrm{mmul2:} &\sum_k A_{ik}B_{jk} \\
  \textrm{ttv:} &\sum_k C_{ijk}v_k \\
  \textrm{ttm:} &\sum_k C_{ijl}A_{kl} \\
  \textrm{mttkrp:} &\sum_{j,k} C_{ijk}A_{jl}B_{kl} \\
  \textrm{inner3:} &\sum_{i,j,k} C_{ijk}C'_{ijk} \\
\end{align*}
  }
\end{minipage}
\qquad
\begin{minipage}[c]{0.75\textwidth}
\begin{lstlisting}
-- Tensor Examples
-- index ordering: i, j, k, l
def mmul1  := Σ j $ (A : i →ₛ j →ₛ R) ⋆ (B : j →ₛ k →ₛ R)
def mmul2  := Σ k $ (A : i →ₛ k →ₛ R) ⋆ (B : j →ₛ k →ₛ R)
def ttv    := Σ k $ (C : i →ₛ j →ₛ k →ₛ R) ⋆ (v : k →ₛ R)
def ttm    := Σ l $ (C : i →ₛ j →ₛ l →ₛ R) ⋆ (A : k →ₛ l →ₛ R)
def mttkrp := Σ j $ Σ k $ (C : i →ₛ j →ₛ k →ₛ R) ⋆
                   (A : j →ₛ l →ₛ R) ⋆ (B : k →ₛ l →ₛ R)
def inner3 := Σ i $ Σ j $ Σ k $
    (C : i →ₛ j →ₛ k →ₛ R) ⋆ (C : i →ₛ j →ₛ k →ₛ R)

-- alternative declaration style:
def M1 : i →ₛ j →ₛ R := A
def M2 : j →ₛ k →ₛ R := B
def mat_mul_alt := Σ j (M1 ⋆ M2)

-- missing index leads to type elaboration error:
def mat_mul_err := Σ l (M1 ⋆ M2)

-- a more informative tensor type
def image_type := row →ₛ col →ₛ channel →ₛ intensity
\end{lstlisting}
\end{minipage}
\caption{Multiplicative example expressions from \cite{kjolstad2017tensor}.
  On the right, Lean code implementing each expression.
  These expressions can be compiled to efficiently executable code.
  Note that \ttt{(\$)} denotes function application and $(I \to_s V)$ denotes the type $\S(I,V)$.
  These examples illustrate the notation for summation, multiplication, and implicit replication.
  The final example fails to type check because the desired summation index is missing from the expressions.
  Index names can be informative and refer to an arbitrary finite set.
  \label{fig:lean-examples}}
\end{figure}

The DSL user writes expressions in a syntax inspired by named-index notation \cite{namedtensornotation}.
Traditional matrix and tensor notation requires users to remember the positional location of
each index they refer to;
in contrast, named index notation requires that each index be given a name, and the relative positions can be forgotten.
Recalling the analogy between array access and function application,
the two approaches can be compared to an \emph{untyped, positional}-argument strategy
versus a \emph{typed, keyword}-argument strategy.
Naming allows replication (usually called broadcasting in array programming~\cite{harris2020array}) to be implicit.
Named indices also convey useful information about the structure of data,
and many reshape transformations required by typical numerics libraries can be omitted.
We give several examples in \cref{fig:lean-examples}.
For our simple expressions from tensor algebra, we use traditional single-letter index names.
We also give one example of an array type with richer index labels.

We implement automatic replication operator insertion using typeclasses that merge the indexing types of \ttt{⋆} operands.
Similarly, the \ttt{map} operators needed to apply summation at the appropriate level are also automatically inserted.
When an attempt is made to sum over a missing index, the expression fails to typecheck.

The extensible syntax of Lean makes writing programs as embedded expressions quite natural, and its dependent type system makes it possible to embed type constraints and inference problems.

\subsection{Code Generation}
To compute results, we need two further ingredients: data format iteration and loop construction.
Our implementation provides several definitions for primitive stream types.
As demonstrated by \citet{chou2018}, many common storage formats can be decomposed by level.
We provide a compositional implementation of the compressed level format, which enables DCSR (two-level sparse matrices)
and arbitrarily deeply nested sparse streams.

To compute the entries of an output, we cannot simply aggregate functions
as the stream evaluation map $\semn{-}$ does.
Instead, we represent this process by
parametrizing the code generating function \ttt{eval} (\cref{fig:lean-loop}) by an abstract \emph{location} at which to accumulate results.
We generate one loop for each index $i_n$ appearing in the input expression.
At each execution of the loop corresponding to $i_n$, this location is specialized with the current value of the index $x \in I_{i_n}.$
The key functions which accomplish this are reproduced fully in the figure.

\subsection{Evaluation}
\begin{figure}
\begin{tabular}{c | c}
kernel & runtime ratio to TACO \\
\hline
mmul1  & 1.12 \\
mmul2  & 0.96 \\
ttv    & 1.03 \\
ttm    & 0.97 \\
mttkrp & 0.88 \\
inner3 & 1.67 \\
\end{tabular}
  \caption{Performance of \ttt{Etch-}generated code relative to TACO. Lower is better.
    \label{fig:perf_data}
  }
\end{figure}

To evaluate expressiveness and performance, we compare to the TACO sparse tensor algebra compiler.
TACO handles a variety of level formats and arbitrary problems written in \emph{concrete index notation}.
TACO generates portable C code.
Its implementation is approximately 25KLOC of C++.

Our prototype compiler handles arbitrary variable contraction expressions.
It consists of less than 200 lines of stream-manipulation code,
plus a simple translator from \ttt{Prog} to C++ and some supporting code for the front end notation.

Although \ttt{Etch} has less generality than TACO at this time, it is already sufficient to reimplement fundamental benchmarks.
\cref{fig:lean-examples} shows the 2\textsuperscript{nd}- and 3\textsuperscript{rd}-order multiplicative
benchmark expressions from \cite{kjolstad2017tensor} alongside their Lean implementations.
We invoke TACO using its preferred sparse level format
and evaluate expressions on synthetic sparse tensor data.
We observe that our system is able to generate structurally equivalent code to TACO.
As a result, the observed performance on most tests is within a factor of 10\% and at most 70\% (\cref{fig:perf_data}).
Our compiler also generates short, readable output code.
The most complex example shown here (MTTKRP) is 38 lines.

\section{Correctness}
\label{sec:proofs}

In this final section, we prove that our stream combinators compute the variables they are expected to compute.
We use the language of {universal algebra} \cite{denecke2009universal} to state our correctness theorem.
Stating the theorem and its proof requires three steps:
we formalize the three contraction operators as a {signature} $\signature$,
present our two operator implementations as separate {algebras} over the signature,
and show that $\semn{-}$ is a {homomorphism} between the algebras.

The first two steps have essentially already been done, so resolving them is just a matter of collecting the details.
The key work occurs in showing that $\semn{-}_\alpha$ (\cref{def:stream_eval_final}) is a homomorphism in \cref{sec:hom}.

\subsection{Universal Algebra}

Universal algebra is a general framework for describing algebraic structures.
Many structure types (such as groups, semirings, and vector spaces) can be described as a collection of atomic types
known as \emph{sorts} together with some operations and identities that they satisfy;
taken together these components form a \emph{signature}.
Each operator is given a type that is parametrized by the sorts it can be applied to.
Here, we give a signature for the variable contraction operators.

\begin{definition}[The Contraction Signature $\signature$]
  Fix an integer $m$.
  \begin{itemize}
  \item
  Each subset of indices $S \subseteq \{1,\ldots, m\}$ denotes a {sort} or \emph{type}.
  The set of sorts is called $\sorts$.
  \item
  $\signature$ contains the following (families of) operators:
  \begin{align*}
    (\cdot) & : S \x S \to S \textrm{ for all $S$} \\
    \sigop_i &: S\cup i \to S \textrm{ for all $S$ not containing $i$} \\
    \repop_i &: S \to S \cup i \textrm{ for all $S$ not containing $i$} \\
  \end{align*}
  \item $\signature$ contains no equalities.
  \end{itemize}
\end{definition}

The purpose of a signature is to specify a collection of symbolic operations.
Any concrete collection of sets and operators that match the types and satisfy the identites is known as an \emph{algebra}.

\begin{definition}
An \emph{algebra} $A$ over the signature $\signature$ consists of a set $A_S$ for each sort $S\in\sorts$
and a function of the appropriate type for each operator:

\begin{itemize}
\item For each $S,$ a function $(\cdot)_S : A_S\x A_S\to A_S$.
\item For each $i, S\notni i,$ a function $\sigop : A_{S\cup i}\to A_S.$
\item For each $i, S\notni i,$ a function $\repop : A_{S}\to A_{S\cup i}.$
\end{itemize}

Each $A_S$ is called a \emph{carrier set}.
\end{definition}

There is a natural notion of \emph{algebra map} that generalizes the various notions of a homomorphism (group homomorphism, linear map, etc.) in algebra.

\begin{definition}
 A \emph{morphism of algebras} or \emph{homomorphism} or \emph{map} between $\signature-$algebras $A, B$, denoted $f : A\to B$,
 is a collection of functions $f_S : A_S\to B_S$ which commute with all operators in the signature
  (subscripts omitted):
  \begin{align*}
  f(a\cdot b) &= f(a)\cdot f(b) \\
  f(\sig{i}{a}) &= \sig{i}{f(a)} \\
  f(\rep{i}{a}) &= \rep{i}{f(a)}.
  \end{align*}
\end{definition}

Given a signature, the \emph{term algebra} is a syntactic algebra: each set is simply composes from well-typed expressions that can be formed using the operators
and a collection of symbolic variables.
This algebra is what we denoted $\L$ in \cref{fig:overview}.
It is what we refer to as the language of contraction expressions.

\begin{definition}[Contraction Signature Term Algebra]
  \label{def:term_algebra}
  Let $X$ be a set; we call the elements \emph{symbolic variables}.
  Let $\tau : X \to \sorts$ be an assignment of sorts to the variables.
  The \emph{term algebra} $\L[X]$ \cref{goguen1977initial} is an algebra over $\signature$ consisting of
  well-typed terms freely assembled from variables and operators in $\signature$.
\end{definition}

\subsection{The Variable Algebra}
The variable algebra encodes the natural interpretation of the variable contraction operators.
It describes how to evaluate an arbitrary contraction expression as a function.
\label{sec:algebra_var}
\begin{definition}[The Variable Algebra $\T$]
  \label{def:variable_algebra}
Assume finite index sets $I_1, \ldots, I_m$ and a semiring $R$.

Let $\pi_i : I_{S\cup \{i\}} \to I_S$ be the projection function defined for $i, S\notni i$.

The variable algebra $\mathcal{T}$ consists of
\begin{align*}
\mathcal{T}_S &= \{ V : I_S \to R \} \\
(V_1 \cdot V_2)(i) &= V_1(i) \cdot V_2(i) \\
(\sig{i}{f}) &= \sum_{x\in I_i}f(x) \\
(\rep{i}{f}) &= f\circ \pi_i.
\end{align*}
This is an unmodified repackaging of the definitions of \cref{sec:semantics}.
\end{definition}

\subsection{The Stream Algebra}

The stream algebra is defined in three steps.
First, we define a well-behaved subset of streams (the \emph{simple streams}).
Then, we define the algebra using the sets and operators introduced in \cref{sec:stream_evaluation}.
Afterwards, in \cref{sec:hom}, we show that the set of simple streams is closed under the operators.
\subsubsection{Simple Streams}
From now on, we are interested in streams $q$ that meet three simple conditions.
The first, \emph{finiteness}, ensures that their evaluation is defined.
The second, \emph{monotonicity}, ensures that they traverse their index sets in the globally defined lexicographic order,
which is needed for efficient multiplication and ensures the ordered traversal condition of fusion (\cref{def:fusion}).
The last technical condition guarantees that multiplication is well-defined.
Formally:

\begin{description}
\item[Finite:] $q$ must reach a terminal state, so $\sem{q}$ is defined.
  Furthermore, the terminal state $t$ must satisfy $\ready(t) = \bot$
  so that $\sem{t} = \semm{t} + \sem{\delta(t)} = 0$.
\item[Monotonic:] For all $r$ reachable from $q$, $$\i(r) \le \i(\delta(r)).$$
\item[Reduced]: If $\reaches{q}{r}$ and $\reaches{r}{s}$, $\ready(r) = \ready(s) = \top$, and $\i(r)=\i(s)$ then $r = s$.
\end{description}

\begin{definition}[Simple Indexed Streams]
  If a stream satisfies all of these properties we call it \emph{simple}.
  Notice that if $q$ is simple, $\delta(q)$ is simple.
  Redefine $\streams{I}{V}$ to denote the set of simple streams of type $I\to V$.
\end{definition}

\remark
Any finite stream with terminal state $t$ such that $\ready(t)=\top$ can be modified in order to satisfy the finiteness condition:
augment the stream with a new state $t'$, $\ready(t')=\bot, \delta(t')=t'$ and $\delta(t) = t'.$
Thus this restriction causes no loss of generality and we do not check it explicitly in our later proofs.
The property is chosen so that all simple streams satisfy the following identity:
$$\sem{q} = \semm{q} + \sem{\delta(q)}.$$

The components of the stream algebra have been defined earlier, so
this definition simply combines \cref{def:nested_streams} and \cref{def:nested_operators}.

\begin{definition}[The Stream Algebra $\S$]
  \label{def:stream_algebra}

Assume finite index sets $I_1, \ldots, I_m$ and a semiring $R$.

The simple stream algebra $\mathcal{S}$ consists of:
\begin{itemize}
  \item Each sort $S\subseteq\numbers{m}$ corresponds to an ordered sequence $\alpha(S)$ (its elements, in order).
    We define $$\S_S = \{q \in \S_{\alpha(S)}\mid q \textrm{ is simple}\}.$$
  \item The operators are $(\cdot),\sigop_i,\repop_i$.
\end{itemize}
It is reasonably straightforward to check that applying any operator to a simple stream yields a simple stream.
This is done in the following operator-specific subsections of \cref{sec:hom}. 
\end{definition}

\subsection{The Correctness Theorem}
In the remainder of the section, we use $\semnop(-)$ as an alternate notation for $\semn{-}.$

Informally, the correctness theorem states that for any collection of streams,
a contraction expression evaluated using stream combinators and then $\semnop$ gives the same result as
first evaluating the streams with $\semnop$ and then applying the variable combinators.
The variable combinators (\cref{def:variable_algebra}), which are simple operations defined on functions, serve as an easy to understand semantics.

The collection of streams is formalized as a set $X$ of symbolic variables with a type assignment
$\tau : X\to T$ mapping each symbol $x$ to a sort $\tau(x)\subseteq\numbers{m}$
and a function $v : X\to \S$ mapping $x$ to a stream in $\S_{\tau(x)}$.
The arbitrary contraction expression is formalized as an element of $\L[X]$.

The result is a simple consequence of the fact that $\semnop$ is a homomorphism of $\signature-$algebras.
This is proved in the following subsection.
We use the following well known fact about term algebras:

\begin{lemma}
  First let $v_\L : X\to \L[X]$ name the function which includes $X$ in the term algebra.
  The term algebra $\L[X]$ is \emph{initial}~\cite{goguen1977initial}:
  for any other $\signature-$algebra $A$ and function $v : X\to A$,
  there is a unique algebra map
  $$\overline{v} : \L[X]\to A$$
  such that $\overline{v}\circ v_\L = v$.
\end{lemma}
This map $\overline{v}$ is often known as an interpretation.
It maps a syntactic term into the domain $A$
by applying each operator that appears to the interpretations of its parts.

Suppose we have a context $v : X\to\S$ assigning streams to symbolic variables.
The two methods of evaluating a contraction expression on these streams
that were mentioned at the beginning of this section are precisely the following two maps:
$$\semnop\circ~\overline{v} : \L[X]\to \T$$
and
$$\overline{\semnop\circ~v} : \L[X]\to \T.$$
\begin{theorem}[Correctness Theorem]
  \label{the:final}
  For all $v : X\to\S$,
  $$\semnop\circ~\overline{v} = \overline{\semnop\circ~v}.$$
\end{theorem}
\begin{proof}
  By \cref{the:hom}, $\semnop$ is a map of algebras;
  hence the left and right hand sides are both maps $\L[X]\to\T$.
  When precomposed with $v_\L$, they both give the same variable assignment $\semnop\circ~v$;
  hence, by the initiality property of $\L[X]$, they must be the same map.
\end{proof}

\subsection{The Correctness Proof}
\label{sec:hom}
In this section, we do most of the work for the correctness result by showing that each stream operation commutes with $\semnop:$
\begin{theorem}
  The function $\semn{-}$ (also written $\semnop$) is a homomorphism of $\signature-$algebras:
  \label{the:hom}
  \begin{align*}
  \semn{a\cdot b}   &= \semn{a}\cdot\semn{b} \\
  \semn{\sig{i}{a}} &= \sig{i}{\semn{a}} \\
  \semn{\rep{i}{a}} &= \rep{i}{\semn{a}}
  \end{align*}
  For all simple streams $a,b$.
\end{theorem}
Notice that on the left, each operator acts on streams, while on the right, each acts on variables.

The proof occupies the remainder of this section.
It is accomplished by proving a series of simple lemmas about the operators used to implement the three primary operators.

In addition, we show that the simple streams are closed under each of the operators, so $\S$ really is an algebra.

\subsubsection{Map Lemmas}

The $\{\sigop,\repop\}$ operators are defined on nested streams using $\map$ over a certain prefix of unaffected indices.
Implicitly, the variable operators are as well.
The following definition clarifies this point, and \cref{lem:map} allows us to ignore this complication when proving correctness:

\begin{definition}[Variable map]
  For $f : \T_\alpha \to \T_\beta$, define $\map_i f : \T_{i::\alpha}\to \T_{i::\beta}$ by
  $$(\map_i f)(V) = f \circ V.$$
  By extension, define
  \begin{align*}
  \map_{[]} &= \id \\
  \map_{i::\gamma} &= \map_i\circ\map_\gamma
  \end{align*}
\end{definition}

This is linked to stream $\map$ by the following identity:
\begin{lemma}
  \label{lem:map_var}
 $$(\map_i f) \sem{q} = \sem{\map_i f~q}.$$
\end{lemma}
\begin{proof}
$$(\map_i f)\sem{q} = f \circ\left(\sum_r(\i(r)\mapsto\v(r))\right) = \sum_r(\i(r)\mapsto f(\v(r))) = \sem{\map_i f~q}.$$
\end{proof}

\begin{lemma}
  \label{lem:map}
  Suppose we have functions $f : \S_\alpha\to\S_\beta$ and $f' : \T_\alpha\to\T_\beta$ that satisfy
  $$\semnop(f(q)) = f'(\semnop(q))$$ for all $q\in\S_\alpha$.
  Then for all $i\notin(\alpha\cup\beta)$ and $q\in \S_{i::\alpha}$,
  $$ \semnop~(\map_i f ~ q) = (\map_i f')(\semnop~q). $$
  By extension, for any sequence of indices $S$ such that the following expressions are well-defined,
  $$ \semnop~(\map_S f ~ q) = (\map_S f')(\semnop~q). $$
\end{lemma}
\begin{proof}
  We insert indices for $\semnop$ that are implicit above:
  \begin{align*}
    \semnop_{i::\beta}(\map_i f~q)
    &= \sem{\map_i(\semnop_\beta\circ f)~q} &\textrm{(definition)} \\
    &= \sem{\map_i(f'\circ \semnop_\alpha)~q} & \textrm{(assumption)} \\
    &= \map_i f'~ \sem{\map_i~\semnop_\alpha~q} &\cref{lem:map_var} \\
    &= \map_i f'~(\semnop_{i::\alpha}~q). &\textrm{(definition)} \\
  \end{align*}
\end{proof}

Having shown this, we now assume that $\alpha = []$ (as in \cref{def:nested_operators})
for the remaining proofs of operator correctness.

\subsubsection{Replication Correctness}
Replication is easy to check: it boils down to the fact that the evaluation of a stream which produces the same value at every index is a constant function.

\begin{theorem}
If $q\in \S$ is simple, then $\rep{i}{q}$ is simple.
\end{theorem}
\begin{proof}\mbox{}
\begin{description}
  \item [Finite] The statespace is $\numbers{k}$ and $k$ is a terminal state, so the result is finite.
  \item [Monotonic] By assumption, the function $\i : \numbers{k}\to I_i$ is order preserving.
  \item [Reduced] By assumption, $\i$ is a bijection, so each state has a unique index value.
\end{description}
\end{proof}
\begin{theorem}
$$\semn{\rep{i}{a}} = \rep{i}{\semn{a}}$$
\end{theorem}
\begin{proof}
  Obvious from the definition of $\repop$ and \cref{lem:map}.
\end{proof}
\subsubsection{Multiplication Correctness}
Consider the stream $(a\cdot b)$.
Since it is simple (as we will show in a moment), it satisfies
\begin{align}
  \label{eval-identity}
  \semn{a\cdot b} = \semnn{a\cdot b} + \semn{\delta(a\cdot b)}.
\end{align}

The key proof intuition is this: since $a$ and $b$ are reduced,
as soon as $\semnn{a}\semnn{b}\neq 0$, we can immediately advance either stream.
If they were not reduced, we might need to multiply a series of terms from both streams;
but monotonicity and reducedness are sufficient to avoid this.

\begin{theorem}
  If $a,b\in \S_\alpha$ then $a\cdot b$ is also simple.
\end{theorem}
\begin{proof}\mbox{}\\
  \begin{description}
  \item[Finite] Since $a$ and $b$ reach a terminal state in a finite number of steps and each transition of $a\cdot b$ advances one or the other, $a\cdot b$ must reach a terminal state as well.
  \item[Monotonic] $a\cdot b$ is monotonic because $\i(a\cdot b) = \max(\i(a), \i(b))$ and $\max$ is monotone in both arguments.
  \item[Reduced] Suppose the state $a'\cdot b'$ is reachable from $a\cdot b$.
    If both states are ready then $\i(a)=\i(b)$ and $\i(a') = \i(b')$.
    If $\i(a'\cdot b') = \i(a\cdot b)$, then in fact $\i(a) =\i(a')$ also, so since $a$ is reduced, $a=a'$.
    Similary $b=b'$, so the states are equal.
  \end{description}
\end{proof}

\begin{theorem}
  For all $a,b\in\streams{I}{A}$,
$$\semn{a\cdot b} = \semn{a}\cdot\semn{b}$$
\end{theorem}
\begin{proof}
We induct over the number of steps to reach a terminal state.
If $a\cdot b$ is terminal, it must be that $a$ and $b$ are both terminal, so $\semn{a\cdot b} = 0 = \semn{a}\semn{b}$.

Otherwise, first suppose $a < b$.
Then $\semnn{a\cdot b} = 0$ and $\semnn{a}\semn{b} = 0$ since $b$ is monotonic.
So
\begin{align*}
\semn{a\cdot b} = \semn{\delta(a)\cdot b} = \semn{\delta a}\semn{b} = (\semnn{a} + \semn{\delta a})\semn{b} = \semn{a}\semn{b},
\end{align*}
and similarly for $b < a$.

Otherwise, we are in the case of $\i(a) = \i(b)$ and $\ready(a)=\ready(b)$.
The interesting case is $\ready(a)=\ready(b)=\top$.
In this case, $\semnn{a\cdot b} = \semnn{a}\semnn{b}.$
Since $b$ is reduced, $\semnn{a}\semn{\delta(b)} = 0;$ otherwise we would have a state $b'\neq b, \reaches{b}{b'}$ with
$\ready(b')$ and $\i(b') = \i(a) = \i(b).$
Thus we calculate
\begin{align*}
 \semn{a\cdot b} &= \semnn{a\cdot b} + \semn{\delta(a\cdot b)} & \cref{eval-identity} \\
&= \semnn{a}\semnn{b} + \semn{\delta(a)\cdot b} & \textrm{(definition)} \\
&= \semnn{a}\semnn{b} + \semn{\delta(a)}\semn{b} & \textrm{(induction)} \\
&= \semnn{a}\semnn{b} + \semnn{a}\semn{\delta(b)} + \semn{\delta(a)}\semn{b} & \textrm{(add zero)}\\
&= (\semnn{a} + \semn{\delta(a)})\semn{b} & \cref{eval-identity} \\
&= \semn{a}\semn{b}. & \cref{eval-identity} \\
\end{align*}

Finally, in the fourth case both streams are \emph{not} ready.
Then $\semnn{a\cdot b} = 0$, we advance $a$, and after finitely many steps reach one of the preceding three cases or a terminal state before emitting anything.
\end{proof}

\paragraph{Note} This proof only requires that one of the two streams be reduced.
Thus, in a compound expression, if we can tolerate the result being non-reduced,
we can support at most one non-reduced input stream and still obtain correct results.
\subsubsection{Summation Correctness}
Summation correctness follows directly from correctness of binary stream addition.

\begin{theorem}
  For all $a,b\in\S_\alpha$, $$\semn{a + b} = \semn{a}+\semn{b},$$
  and $a+b$ is simple.
\end{theorem}
\begin{proof}
  The proof follows essentially the same argument as multiplication.
\end{proof}
\label{sec:summation-correct}

\begin{theorem}
  For all $q\in\streams{I}{V},$
  $$\semn{\sigop{q}} = \sigop{\semn{q}}.$$
\end{theorem}
\begin{proof}

  \begin{align*}
  \semn{\sigop{q}}
    &= \left\llbracket \sum_{\reaches{q}{r}} \v(r)\right\rrbracket & \textrm{definition of $\sigop$} \\\\
    &= \sum_{\reaches{q}{r}}\semn{\v(r)} & \semn{a+b} = \semn{a}+\semn{b} \\
    &= \sum_{\reaches{q}{r}}(\i(r) \mapsto \semn{\v(r)})(\i(r)) &  \textrm{} \\
    &= \sum_{x\in I_i}\left(\sum_{\reaches{q}{r}}\i(r) \mapsto \semn{\v(r)}\right)(x) &  (*) \\
    &= \sum_{x\in I_i}\semn{q}(x)= \sigop{\semn{q}}. & \textrm{def. $\semn{-}$; def. $\sigop$ on variables} \\
  \end{align*}

  Step $(*)$ follows because each inner summand is non-zero for at most one value of $x$,
  so each step of the inner loop is picked out by some step of the outer loop.
  The general claim for $\sigop_i$ follows from \cref{lem:map}.
\end{proof}

\begin{theorem}
  If $q\in \S_S$ is simple, $\sigop_i(q)$ is simple
\end{theorem}
\begin{proof} \mbox{}
  Follows from simplicity of stream addition.
\end{proof}
This completes the proof that $\S$ is an algebra and that $\semnop : \S\to\T$ is a map of algebras.

\section{Related Work}

This paper proposes an operational semantics for generalized contractions and a correct-by-construction DSL compiler.
We discuss prior work on generalized contractions, related DSL systems and compilers, and work on verified compilation.

\paragraph{Generalized Variable Contraction Formulations}

Sparse tensor algebra, databases~\cite{shaikhha2018push}, factorized probability distributions~\cite{aji2000generalized}, weighted graphs~\cite{mattson2013standards}, and formal languages~\cite{Elliott2019-convolution-extended} can all be represented as vectors or higher-rank sparse tensors by choosing the underlying set of scalars appropriately, and such representations are conducive to algebraic restatements of many algorithms~\cite{abo2016faq}.
Moreover, these restatements can enable application of specialized fusion techniques such as worst-case optimal join methods~ \cite{ngo2018worst,veldhuizen2012leapfrog,schleich2019layered} and factorization techniques that perform asymptotically better on some queries.
Thus, formalisms that can uniformly represent algorithms that traverse these objects in streaming fashion have been shown to be useful in the design of optimizations and compiler backends to accelerate problems across new domains.

\paragraph{Compilers and Execution Systems}

Researchers have built systems and compilers for executing several computational languages that are sub-languages of the generalized variable contractions, including tensor algebra, relational algebra, and graph computations.

One line of recent work showed how to compile~\cite{kjolstad2017tensor,tian2021high,mlirsparse} and optimize~\cite{senanayake2020,kjolstad2019} arbitrary sparse tensor algebra expressions to fused code on several types of sparse and dense data structures~\cite{chou2018}. Moreover, Henry and Hsu et al.~\cite{henry2021compilation} showed how to compile general sparse array programs. Our work, however, generalizes these compilers to generalized contractions, including relations and graph computations. Moreover, our work provides a formal foundation for these systems and a proof of correctness, while generating equivalent code to them.

Execution systems for relational algebra~\cite{codd1971relational} as used in database management systems has flourished since System R~\cite{astrahan1976system} and INGRES~\cite{held1975ingres}.
Several libraries also provide relational algebra support, including Python pandas~\cite{panads} and SQLite~\cite{sqlite}.
Recently, Aberger et al. showed how to generate fused code for inner join expressions~\cite{aberger2017emptyheaded}.
Finally, researchers have shown how to combine dense tensor algebra with relational algebra~\cite{aberger2018levelheaded,yuan2020tensor}.
Our work is more general and can handle both dense and sparse tensor algebra, relational algebra, and more.
Furthermore, we provide a compiler approach that can generate bespoke fused code for general queries beyond inner joins.

Finally, programming systems for graph computing have become popular over the last decade~\cite{kulkarni2007optimistic,low2014graphlab,malewicz2010pregel,zhang2018graphit,shun2013ligra}.
These are typically programmed in a graph-based abstraction, but many of the algorithms they are used to implement can also be expressed as generalized contractions~\cite{kepner2011graph}.
Based on this observation, the GraphBLAS standard~\cite{kepner2016mathematical} was developed to express graphs algorithms in the language of linear algebra.
Unlike these graph systems, our work provides a formal model for the fused execution of graph algorithms expressed as generalized contractions as well as a correctness proof and a compiler that generates fused code.

\paragraph{Mechanical Verification of High Performance Systems}

There is much prior work on certified compilation of low-level languages.
Here we mention two examples:
CompCert~\cite{leroy2016compcert} is a certified, monolithic C compiler,
and Bedrock~\cite{chlipala2013bedrock} is an extensible system for verified systems programming.
These tools assist the programmer in building and verifying high-performance systems,
but focus on a lower level of abstraction where sparse array optimizations are difficult to express.

Other work has developed formal models and mechanized correctness proofs for \emph{dense} linear algebra compilation
~\cite{reinking2020formal, verpoly}.
These works focus on proving the correctness of various sophisticated optimization strategies for polyhedral programs.
These methods achieve high performance on dense problems, but
they do not apply to compressed data structures used to represent sparse data.
The present work is orthogonal in that it handles optimizations that are unique to compressed data structures,
but it can also express the dense iteration pattern.
For example, a compiler built according to stream semantics could make use of externally generated and verified polyhedral streams for dense sub-problems.

\citet{liu} show that it is possible to express many low-level optimizations on dense array programs
via verified source-to-source transformations of a high level functional language.
Moreover, they show that features of modern interactive proof assistants can be used to improve the productivity of algorithm designers and optimizers.
Our work shares the point of view that numeric computation systems can be made simultaneously simpler, more trustworthy, and more productive
via careful redesign of representations and optimization methods.
In contrast, we address sparse data structures and emphasize the variable contraction point of view.

The notion of modeling arrays as functions has been used in much prior work ~\cite{ragan2013halide, paszke2021getting} to support
optimized compilation of array programs.
The fact that natural join, which is essential to the expressivity of relational algebra,
can be redefined in terms of simple replication and element-wise multiplication (in this case, set intersection)
has been noted before~\cite{imielinski}.

\paragraph{Stream Programs and Stream Fusion}
Stream-fusion~\cite{kiselyov2017stream, coutts2007stream} and one-dimensional stream-based programming models~\cite{halbwachs1991synchronous, thies2002streamit} have been
an important topic in the functional programming community and used in a wide range of applications from embedded signal-processing to database query evaluation.
Our work describes higher-dimensional streams that are augmented with additional indexing parameters.
These indices have semantic content beyond being a proxy for time.
Indexed streams are related to (finite, hierarchical) maps in the same way that standard streams are related to lists,
which enables new composition methods.

\section{Conclusion}

We introduced the indexed stream formal operational semantics for the fused execution of variable contraction expressions.
Since the model hides details of fusion and sparse data structure iteration beneath a high-level functional expression language,
a programmer can focus on the computation they want and a provably correct programming system can handle the rest.
We hope that our indexed stream semantics will enable future certified compilers
for important computations across a wide variety of domains.

\begin{acks}
  We would like to thank
  Manya Bansal,
  Olivia Hsu,
  Matthew Sotoudeh,
  Shiv Sundram,
  Rohan Yadav,
  and Bobby Yan
  for their helpful feedback on earlier drafts.
  We would also like to thank Kyle Miller for his many patient explanations of Lean techniques.

  This work was in part supported by the National Science Foundation under Grant No. CCF-2216964.

\end{acks}

\bibliography{references}

\end{document}